\documentclass[runningheads]{llncs}

\usepackage[utf8]{inputenc}
\usepackage{amsmath}
\usepackage{amsfonts}
\usepackage{graphicx}
\usepackage{xspace}
\usepackage{todonotes}
\usepackage{mathtools}
\DeclarePairedDelimiter\ceil{\lceil}{\rceil}
\usepackage{cite}
\usepackage{microtype}
\usepackage{caption, subcaption}
\usepackage{algorithm}
\usepackage[]{algpseudocode}
\usepackage{textcomp, mathcomp}
\usepackage{array}
\usepackage{comment}
\usepackage{amssymb}
\usepackage{ntheorem}
\DeclareMathOperator{\dist}{dist}

\usepackage{tikz}
\usetikzlibrary{decorations.pathreplacing}

\newcommand{\greg}[1]{{\color{red} {\bf Greg:} #1}}

\newcommand{\Oish}{\widetilde{O}}

\newcommand{\eps}{\varepsilon}

\begin{document}

\title{Weighted Additive Spanners}
\author{Reyan~Ahmed\inst{1} \and
Greg~Bodwin\inst{2} \and 
Faryad~Darabi~Sahneh\inst{1} \and 
Stephen~Kobourov\inst{1} \and 
Richard~Spence\inst{1}}
\institute{Department of Computer Science, University of Arizona \and
Department of Computer Science, Georgia Institute of Technology}
\date{}
\authorrunning{Ahmed et al.}

\maketitle

\begin{abstract}
A \emph{spanner} of a graph $G$ is a subgraph $H$ that approximately preserves shortest path distances in $G$.
Spanners are commonly applied to compress computation on metric spaces corresponding to weighted input graphs.
Classic spanner constructions can seamlessly handle edge weights, so long as error is measured \emph{multiplicatively}.
In this work, we investigate whether one can similarly extend constructions of spanners with purely \emph{additive} error to weighted graphs. These extensions are not immediate, due to a key lemma about the size of shortest path neighborhoods that fails for weighted graphs. Despite this, we recover a suitable amortized version, which lets us prove direct extensions of classic $+2$ and $+4$ unweighted spanners (both all-pairs and pairwise) to $+2W$ and $+4W$ weighted spanners, where $W$ is the maximum edge weight.  Specifically, we show that a weighted graph $G$ contains all-pairs (pairwise) $+2W$ and $+4W$ weighted spanners of size $O(n^{3/2})$ and $\Oish(n^{7/5})$ ($O(np^{1/3})$ and $O(np^{2/7})$) respectively. For a technical reason, the $+6$ unweighted spanner becomes a $+8W$ weighted spanner; closing this error gap is an interesting remaining open problem. That is, we show that $G$ contains all-pairs (pairwise) $+8W$ weighted spanners of size $O(n^{4/3})$ ($O(np^{1/4})$).

\end{abstract}

\section{Introduction}
An \emph{$f(\cdot)$-spanner} of an undirected graph $G=(V,E)$ with $|V|=n$ nodes and $|E|=m$ edges is a subgraph $H$ which preserves pairwise distances in $G$ up to some error prescribed by $f$; that is,
$$\dist_H(s, t) \le f(\dist_G(s, t)) \ \text{for all nodes } s,t \in V.$$
Spanners were introduced by Peleg and Sch\"{a}ffer~\cite{doi:10.1002/jgt.3190130114} in the setting with multiplicative error of type $f(d) = cd$ for some positive constant $c$.
This setting was quickly resolved, with matching upper and lower bounds \cite{Alth90} on the sparsity of a spanner that can be achieved in general. At the other extreme are (purely) \emph{$c$-additive spanners} (or $+c$ spanners), with error of type $f(d) = d+c$. More generally, if $f(d)=\alpha d + \beta$, we say that $H$ is an \emph{$(\alpha,\beta)$-spanner}.
Intuitively, additive error is much stronger than multiplicative error; most applications involve shrinking enormous input graphs that are too large to analyze directly, and so it is appealing to avoid error that scales with distance.

Additive spanners were thus initially considered perhaps too good to be true, and they were discovered only for particular classes of input graphs \cite{liestmannp}.
However, in a surprise to the area, a seminal paper of Aingworth, Chekuri, Indyk, and Motwani \cite{Aingworth99fast} proved that nontrivial additive spanners actually exist \emph{in general}: every $n$-node undirected unweighted graph has a $2$-additive spanner on $O(n^{3/2})$ edges.
Subsequently, more interesting constructions of additive spanners were found: there are $4$-additive spanners on $\Oish(n^{7/5})$ edges \cite{chechik2013new} and $6$-additive spanners on $O(n^{4/3})$ edges \cite{baswana2010additive, knudsen2014additive}.
There are also natural generalizations of these results to the \emph{pairwise} setting, where one is given $G = (V, E)$ and a set of demand pairs $P \subseteq V \times V$, where only distances between node pairs $(s, t) \in P$ need to be approximately preserved in the spanner \cite{Kavitha15, Kavitha2017, bodwin2017linear, coppersmith2006sparse, Bodwin:2016:BDP:2884435.2884496, Hsien2018nearoptimal}.

Despite the inherent advantages of additive error, multiplicative spanners have remained the more well-known and well-applied concept elsewhere in computer science.
There seem to be two reasons for this:
\begin{enumerate}
    \item Abboud and Bodwin \cite{abboud20174} (see also \cite{HP18swat}) give examples of graphs that have no $c$-additive spanner on $O(n^{4/3 - \eps})$ edges, for any constants $c, \eps > 0$.
    Some applications call for a spanner on a near-linear number of edges, say $O(n^{1+\eps})$, and hence these must abandon additive error if they need theoretical guarantees for every possible input graph.
    However, there is some evidence that many graphs of interest bypass this barrier; e.g.\ graphs with good expansion or girth properties \cite{baswana2010additive}.
    
    \item Spanners are often used to compress metric spaces that correspond to \emph{weighted} input graphs.
    This includes popular applications in robotics \cite{MB13, DB14, DBLP:journals/ijcga/CaiK97, SSAH14}, asynchronous protocol design \cite{PU89jacm}, etc.,  and it incorporates the extremely well-studied case of Euclidean spaces which have their own suite of applications (see book \cite{Narasimhan:2007:GSN:1208237}).
    Current constructions of multiplicative spanners can handle edge weights without issue, but purely additive spanners are known for unweighted input graphs only.
\end{enumerate}

Addressing both of these points, Elkin et al.~\cite{elkin2019almost} (following \cite{elkin2005computing}) recently provided constructions of \emph{near-additive} spanners for weighted graphs.
That is, for any fixed $\eps, t > 0$, every $n$-node graph $G=(V, E, w)$ has a $(1+\varepsilon, O(W))$-spanner on $O(n^{1+1/t})$ edges, where $W$ is the maximum edge weight.\footnote{Their result is actually a little stronger: $W$ can be the maximum edge weight on the shortest path between the nodes being considered.}
This extends a classic unweighted spanner construction of Elkin and Peleg \cite{Elkin:2004:SCG:976327.984900} to the weighted setting.
Additionally, while not explicitly stated in their paper, their method can be adapted to a $+2W$ purely additive spanner on $O(n^{3/2})$ edges (extending \cite{Aingworth99fast}).

The goal of this paper is to investigate whether or not all the other constructions of spanners with purely additive error extend similarly to weighted input graphs.
As we will discuss shortly, there is a significant barrier to a direct extension of the method from \cite{elkin2019almost}.
However, we prove that this barrier can be overcome with some additional technical effort, thus leading to the following constructions.
In these theorem statements, all edges have (not necessarily integer) edge weights in $(0,W]$. Let $p=|P|$ denote the number of demand pairs and $n=|V|$ the number of nodes in $G$.

\begin{theorem} \label{thm:intro2span}
For any $G=(V, E, w)$ and demand pairs $P$, there is a $+2W$ pairwise spanner with $O(np^{1/3})$ edges.
\end{theorem}

\begin{theorem} \label{thm:intro4span}
For any $G=(V, E, w)$ and demand pairs $P$, there is a $+4W$ pairwise spanner with $O(np^{2/7})$ edges.
In the all-pairs setting $P = V \times V$, the bound improves to $\Oish(n^{7/5})$.
\end{theorem}

These two results exactly match previous ones for unweighted graphs \cite{chechik2013new, Kavitha2017, Kavitha15, Aingworth99fast}, with $+2W$ ($+4W$) in place of $+2$ ($+4$).
Theorem \ref{thm:intro2span} is partially tight in the following sense: it implies that $O(n^{3/2})$ edges are needed for a $+2W$ spanner when $p=O(n^{3/2})$, and neither of these values can be unilaterally improved.
Our next two results are actually a bit weaker than the corresponding unweighted ones \cite{Kavitha2017, Cygan13, Pettie09}: for a technical reason, we take on slightly more error in the weighted setting (the corresponding unweighted results have $+6$ and $+2$ error respectively).

\begin{theorem}
For any $G =(V,E,w)$ and demand pairs $P$, there is a $+8W$ pairwise spanner with $O(np^{1/4})$ edges.
In the all-pairs setting $P = V \times V$, the bound improves to $O(n^{4/3})$.
\end{theorem}

\begin{theorem} \label{thm:intross}
For any $G = (V,E,w)$ and demand pairs $P = S \times S$, there is a $+4W$ pairwise spanner with $O(n|S|^{1/2})$ edges.
\end{theorem}

We summarize our main results (except Theorem \ref{thm:intross}) in Table~\ref{TABLE:results}, contrasted with known results for unweighted graphs. 

\begin{table}[t]
\centering
\setlength{\tabcolsep}{10pt}
\begin{tabular}{|c|l|c|l|}
\hline
\multicolumn{2}{|c|}{Unweighted} & \multicolumn{2}{c|}{Weighted} \\ \hline
Stretch & \multicolumn{1}{c|}{Size} & Stretch & \multicolumn{1}{c|}{Size} \\ \hline
$+2$ & $O(n^{3/2})$~\cite{Aingworth99fast} & $+2W$ & $O(n^{3/2})$ \cite{elkin2019almost} \\ \hline
$+4$ & $\Oish(n^{7/5})$~\cite{chechik2013new} & $+4W$ & $\Oish(n^{7/5})$ [\textbf{this paper}] \\ \hline
$+6$ & $O(n^{4/3})$~\cite{baswana2010additive, knudsen2014additive, woodruff2010additive} & $+6W$ & ? \\ \hline
$+c$ & $\Omega(n^{4/3 - \eps})$~\cite{abboud20174, HP18swat} & $+8W$ & $O(n^{4/3})$ [\textbf{this paper}] \\ \hline
\end{tabular}

\caption{Table of additive spanner constructions for unweighted and weighted graphs, where $W$ denotes the maximum edge weight. \label{TABLE:results}}
\end{table}

\subsection{Technical Overview: What's Harder With Weights?}

There is a key point of failure in the known constructions of unweighted additive spanners when one attempts the natural extension to weighted graphs.
To explain, let us give some technical background.
Nearly all spanner constructions start with a \emph{clustering} or \emph{initialization} step: taking the latter exposition \cite{knudsen2014additive}, a \emph{$d$-initialization} of a graph $G$ is a subgraph $H$ obtained by choosing $d$ arbitrary edges incident to each node, or all incident edges to a node of degree less than $d$.
After this, many additive spanner constructions leverage the following key fact (the one notable exception is the $+2$ all-pairs spanner, which is why one can recover the corresponding weighted version from prior work):
\begin{lemma} [\hspace{1sp}\cite{Cygan13, Kavitha15, Kavitha2017, chechik2013new}, etc.] \label{lem:pathnbhd}
Let $G$ be an undirected unweighted graph, let $\pi$ be a shortest path, and let $H$ be a $d$-initialization of $G$.
If $\pi$ is missing $\ell$ edges in $H$, then there are $\Omega(d\ell)$ different nodes adjacent to $\pi$ in $H$.
\end{lemma}
\begin{proof}
For each missing edge $(u, v) \in \pi$, by construction both $u$ and $v$ have degree at least $d$ in $H$ (otherwise, $\text{deg}_H(u) < d$, in which edge $(u,v)$ is added in the $d$-initialization $H$).
By the triangle inequality, any given node is adjacent to at most three nodes in $\pi$.
Hence, adding together the $\ge d$ neighbors of each of the $\ell$ missing edges, we count each node at most three times so the number of nodes adjacent to $\pi$ is still $\Omega(d\ell)$.
\qed
\end{proof}

The difficulty of the weighted setting is largely captured by the fact that Lemma~\ref{lem:pathnbhd} fails when $G$ is edge-weighted.
As a counterexample, let $\pi$ be a shortest path consisting of $\ell+1$ nodes and $\ell$ edges of weight $\eps$. Additionally, consider $d$ nodes, each connected to every node along $\pi$ with an edge of weight $W > \eps \ell$. A candidate $d$-initialization $H$ consists of selecting every edge of weight $W$. In this case, all $\ell$ edges in $\pi$ are missing in $H$, but there are still only $d \neq \Omega(d\ell)$ nodes adjacent to $\pi$ in $H$.

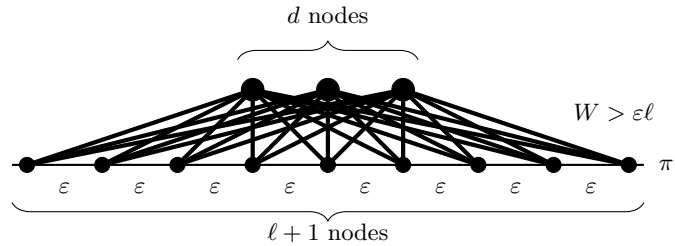
\begin{figure}
    \centering
    \begin{tikzpicture}
\foreach \x in {0,...,2}{
    \draw [fill=black] (\x, 1) circle [radius=0.15];
    \foreach \y in {-3,...,5}{
        \draw [ultra thick] (\x, 1) -- (\y, 0);
    }
}

\draw [decorate,decoration={brace,amplitude=7}] (-0.2, 1.4) -- (2.2, 1.4);
\node at (1, 2) {$d$ nodes};

\draw [decorate,decoration={brace,amplitude=7}] (5.2, -0.5) -- (-3.2, -0.5);
\node at (1, -0.9) {$\ell+1$ nodes};

\node at (4.8, 0.7) {\bf $W > \eps \ell$};

\foreach \x in {-3,...,5}{
    \draw [fill=black] (\x, 0) circle [radius=0.1];
}
\foreach \x in {-3,...,4}{
    \node at ({\x+0.5}, -0.3) {$\eps$};
}

\draw [thick] (-3.2, 0) -- (5.2, 0);
\node at (5.5, 0) {$\pi$};
    \end{tikzpicture}
    \caption{A counterexample to Lemma \ref{lem:pathnbhd} for weighted graphs.}
    \label{fig:cenbhd}
\end{figure}

The fix, as it turns out, is simple in construction but involved in proof.
We simply replace initialization with \emph{light initialization}, where one must specifically add the lightest $d$ edges incident to each node.
With this, the proof of Lemma \ref{lem:pathnbhd} is still not trivial: it remains possible that an external node can be adjacent to arbitrarily many nodes along $\pi$, so a direct counting argument fails.
However, we show that such occurrences can essentially be amortized against the rising and falling pattern of missing edge weights along $\pi$.
This leads to a proof that \emph{on average} an external node is adjacent to $O(1)$ nodes in $\pi$, which is good enough to push the proof through.
We consider this weighted extension of Lemma \ref{lem:pathnbhd} to be the main technical contribution of this work, and we are hopeful that it may be of independent interest as a structural fact about shortest paths in weighted graphs.

\section{Neighborhoods of Weighted Shortest Paths}

Here we introduce the extension of Lemma \ref{lem:pathnbhd}.
Following the technique in \cite{knudsen2014additive}, define a \emph{$d$-light initialization} of a weighted graph $G = (V, E, w)$ to be a subgraph $H$ obtained by including the $d$ lightest edges incident to each node (or all edges incident to a node of degree less than $d$). Ties between edges of equal weight are broken arbitrarily; for clarity we assume this occurs in the background so that we can unambiguously refer to ``the lightest $d$ edges'' incident to a node. We prove the weighted analogue of Lemma~\ref{lem:pathnbhd}.

\begin{theorem}
\label{thm:d_initialize}
If $H$ is a $d$-light initialization 
of an undirected weighted graph $G$, and
there is a shortest path $\pi$ in $G$ that is missing $\ell$ edges in $H$, then there are $\Omega(d\ell)$ nodes adjacent to $\pi$ in $H$.
\end{theorem}

We give some definitions and notation which will be useful in the proof of Theorem~\ref{thm:d_initialize}.
Let $s$ and $t$ be the endpoints of a shortest path $\pi$, and let
$M := \pi \setminus E(H)$
be the set of edges in $\pi$ currently missing in $H$ so that $|M| = \ell$.
For convenience we consider these edges to be \textit{oriented} from $s$ to $t$, so we write $(u, v) \in M$ to mean that $\dist_G(s, u) < \dist_G(s, v)$ and $\dist_G(u, t) > \dist_G(v, t)$. Suppose the edges in $M$ are labeled in order $e_1$, $e_2$, \ldots, $e_{\ell}$ where $e_i = (u_i,v_i)$, and let $w_i$ denote the weight of edge $e_i$. Given $u \in V$, let $N^*(u)$ denote the \emph{$d$-neighborhood} of $u$ as follows:
$$N^*(u) := \left\{v \in V \ \mid \ (u, v) \text{ is one of the lightest $d$ edges incident to $u$} \right\}.$$
We will show that the size of the union of the $d$-neighborhoods of the nodes $u_1$, \ldots, $u_\ell$ is $\Omega(d\ell)$, that is
$$\left| \bigcup \limits_{(u, v) \in M} N^*(u) \right| = \Omega(d\ell)$$
noting that the above set is a subset of all nodes adjacent to $\pi$. 
In particular, the above set may not contain nodes $v'$ connected to $u \in \pi$ by an edge that is 1) among the $d$ lightest incident to $v'$, 2) \emph{not} among the $d$ lightest incident to $u$.
However, the above set necessarily contains all nodes $v'$ which are connected to some $u_i$ or $v_i$ by an edge among the $d$ lightest incident to $u_i$ or $v_i$. We remark that if the $d$-neighborhoods $N^*(u_1)$, $N^*(u_2)$, \ldots, $N^*(u_{\ell})$ are pairwise disjoint, then $|\bigcup_{(u,v)\in M} N^*(u)| = d\ell$, which immediately implies there are at least $d\ell$ nodes adjacent to $\pi$ in $H$. Hence for the remainder of the proof, we assume there exist $i$ and $k$ with $1\le i<k \le \ell$ such that $N^*(u_i) \cap N^*(u_k)$ is nonempty. We use the convention that if $a$ and $b$ are integers with $b<a$, then $\sum_{i=a}^b f(i) = 0$. The following lemma holds (see Figure~\ref{fig:weightineq}):

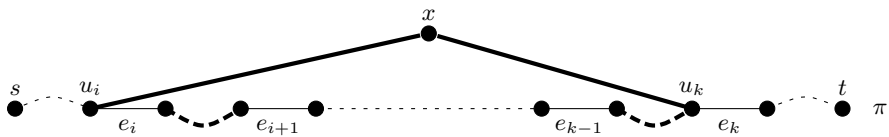
\begin{figure}[h]
\centering
\begin{tikzpicture}
	\node (s) at (-4,0) {};
	\node (t) at (7,0) {};
	\node (x) at (1.5,1) {};
	\draw [fill=black] (x) circle [radius=0.1];
	\draw [fill=black] (s) circle [radius=0.1];
	\draw [fill=black] (t) circle [radius=0.1];
	\foreach \x in {-3,-2,-1,0,3,4,5,6}{
		\draw [fill=black] (\x, 0) circle [radius=0.1];
	}
	\draw [ultra thick] (-3,0) -- (x);
	\draw [ultra thick] (x) -- (5,0);
	\foreach \x in {-3,-1,3,5}{
		\draw (\x,0) -- (\x+1,0);
	}
	\draw [dash pattern={on 1.5pt off 3pt}] (0,0) -- (3,0);
	\draw [ultra thick, dash pattern={on 4pt off 2pt}] (-2,0) .. controls (-1.5,-0.3) .. (-1,0);
	\draw [ultra thick, dash pattern={on 4pt off 2pt}] (4,0) .. controls (4.5,-0.3) .. (5,0);
	\node[above=2pt] at (x) {$x$};
	\node[above=2pt] at (-3,0) {$u_i$};
	\node[above=2pt] at (5,0) {$u_k$};
	\node[below=1pt] at (-2.5,0) {$e_i$};
	\node[below=1pt] at (-0.5,0) {$e_{i+1}$};
	\node[below=1pt] at (3.5,0) {$e_{k-1}$};
	\node[below=1pt] at (5.5,0) {$e_k$};
	\node[above=2pt] at (s) {$s$};
	\node[above=2pt] at (t) {$t$};
	\node at (7.5,0) {$\pi$};
	\draw [dash pattern={on 1.5pt off 3pt}] (s) .. controls (-3.5,0.2) .. (-3,0);
	\draw [dash pattern={on 1.5pt off 3pt}] (6,0) .. controls (6.5,0.2) .. (t);

\end{tikzpicture}
\caption{Illustration of Lemma~\ref{lem:weightineq}. The bold dashed curves represent subpaths in $H$.}
\label{fig:weightineq}
\end{figure}

\begin{lemma}\label{lem:weightineq}Let $\pi$ be a shortest path, let $x\in V$ be a node such that $x \in N^*(u_i) \cap N^*(u_k)$ for some $1\le i<k \le \ell$, and consider the edges $e_i, \ldots, e_k \in M$ with weights $w_i, \ldots, w_k$. Then
$$w_k \ge \sum \limits_{i'=i+1}^{k-1} w_{i'}.$$
\end{lemma}
\begin{proof}
Consider the subpath of $\pi$ from $u_i$ to $u_k$, denoted $\pi[u_i \leadsto u_k]$. We have
\begin{align*}
\sum \limits_{i'=i}^{k-1} w_{i'} &\le \text{length}\left(\pi[u_i \leadsto u_k]\right)\\
&\le w(u_i, x) + w(x, u_k) \tag*{($\pi[u_i \leadsto u_k]$ is a shortest path)}\\
&\le w_i + w_k
\end{align*}
where the last inequality follows from the fact that edges $(u_i, x), (x, u_k)$ are among the $d$ lightest edges incident to $u_i$ and $u_k$ respectively (since $x \in N^*(u_i) \cap N^*(u_k)$), but $e_i$ and $e_k$ are not, since they are omitted from $H$.
Lemma~\ref{lem:weightineq} follows by subtracting $w_i$ from both sides of the above inequality.
\qed
\end{proof}

For the next part, for edge $e \in M$, say that $e$ is \emph{pre-heavy} if its weight is strictly greater than the preceding edge in $M$, and/or \emph{post-heavy} if its weight is strictly greater than the following edge in $M$. For notational convenience, if an edge is not pre-heavy, we say the edge is pre-light. Similarly, if an edge is not post-heavy, we say the edge is post-light. By convention, the first edge $e_1 \in M$ is pre-light and the last edge $e_{\ell} \in M$ is post-light. We state the following simple lemma; recall that $|M|=\ell$.

\begin{lemma}\label{lem:preheavy}
Either more than $\dfrac{\ell}{2}$ edges in $M$ are pre-light, or more than $\dfrac{\ell}{2}$ edges in $M$ are post-light.
\end{lemma}
\begin{proof}
Let $S_1$ be the set of edges in $M$ which are pre-light, and let $S_2$ be the set of edges in $M$ which are post-light. Note that $e_1 \in S_1$ and $e_{\ell} \in S_2$. For each of the $\ell-1$ pairs of consecutive edges $(e_i, e_{i+1})$ in $M$ where $i = 1, \ldots, \ell-1$, it is immediate by definition that either $e_i \in S_2$ or $e_{i+1} \in S_1$ (or both if $w_i=w_{i+1}$). These statements imply $|S_1|+|S_2| \ge \ell+1$, so at least one of $S_1$ or $S_2$ has cardinality at least $\frac{\ell+1}{2} > \frac{\ell}{2}$.
\qed
\end{proof}

In the sequel, we assume without loss of generality that more than $\frac{\ell}{2}$ edges in $M$ are pre-light; the other case is symmetric by exchanging the endpoints $s$ and $t$ of $\pi$.
We can now say the point of the previous two lemmas: together, they imply that \emph{most} edges $(u, v) \in M$ have mostly non-overlapping $d$-neighborhoods $N^*(u)$. That is:
\begin{lemma}\label{lem:xdiff}
Let $\pi$ be a shortest path. For any node $x \in V$, there exist at most three nodes $u$ along $\pi$ such that $x \in N^*(u)$ and edge $(u,v) \in M$ is pre-light.
\end{lemma}
\begin{proof}
Suppose for sake of contradiction there exist four nodes $u_i$, $u_a$, $u_b$, $u_k$ with $1 \le i < a < b < k \le \ell$ such that $x$ belongs to the $d$-neighborhoods of $u_i$, $u_a$, $u_b$, and $u_k$, and the edges $(u_i,v_i)$, $(u_a,v_a)$, $(u_b,v_b)$, and $(u_k,v_k)$ are pre-light. In particular, we have $k \ge i+3$ and $x \in N^*(u_i) \cap N^*(u_k)$. By Lemma~\ref{lem:weightineq} and the above observation, we have
\[ w_k \ge \sum_{i'=i+1}^{k-1} w_i' = w_{i+1} + \ldots + w_{k-1} \ge w_{i+1} + w_{k-1} \]
By assumption, $e_k = (u_k,v_k)$ is pre-light, so $w_{k-1} \ge w_k$, and the above inequality implies $w_k \ge w_{i+1}+w_{k-1} \ge w_{i+1} + w_k$, or $w_{i+1}=0$. Since edge weights are strictly positive, we have contradiction, proving Lemma~\ref{lem:xdiff}.
\qed
\end{proof}

Finally, define set $X^*$ as follows:
\[X^* := \bigcup \limits_{\substack{(u, v) \in M \\ \text{ is pre-light}}} N^*(u).\]

By Lemma~\ref{lem:preheavy} and the above pre-heavy assumption, there are more than $\frac{\ell}{2}$ pre-light edges $(u,v)$, so the multiset containing all $d$-neighborhoods $N^*(u)$ contains more than $\frac{d\ell}{2}$ nodes. By Lemma~\ref{lem:xdiff}, any given node is contained in at most three of these $d$-neighborhoods, implying $|X^*| > \dfrac{d\ell}{6}$. Since $X^*$ is a subset of $|\bigcup_{(u,v)\in M} N^*(u)|$, we conclude that there are $\Omega(d\ell)$ nodes adjacent to $\pi$ in $H$.
proving Theorem \ref{thm:d_initialize}.

\section{Spanner Constructions}

We show how Theorem~\ref{thm:d_initialize} can be used to construct additive spanners on edge-weighted graphs.
These constructions are not significant departures from prior work; the main difference is applying Theorem~\ref{thm:d_initialize} in the right place.

\subsection{Subset and Pairwise Spanners}

\begin{definition} [Pairwise/Subset Additive Spanners]
Given a graph $G=(V,E,w)$ and a set of demand pairs $P \subseteq V \times V$, a subgraph $H = (V, E_H \subseteq E, w)$ is a \emph{$+c$ pairwise spanner} of $G, P$ if
$$\dist_H(s, t) \le \dist_G(s, t) + c \text{ for all }(s, t) \in P.$$
When $P = S \times S$ for some $S \subseteq V$, we say that $H$ is a \emph{$+c$ subset spanner} of $G, S$.
\end{definition}

In the following results, all graphs $G$ are undirected and connected with (not necessarily integer) edge weights in the interval $(0, W]$, where $W$ is the maximum edge weight. Let $|V|=n$, let $p=|P|$ denote the number of demand pairs (for pairwise spanners), and let $\sigma=|S|$ denote the number of sources (for subset spanners).

\begin{theorem} \label{thm:subspan}
Any $n$-node graph $G=(V,E,w)$ with source nodes $S \subseteq V$ has a $+4W$ subset spanner with $O(n\sigma^{1/2})$ edges.
\end{theorem}
\begin{proof}
The construction of the $+4W$ subset spanner $H$ is as follows, essentially following~\cite{knudsen2014additive}.
Let $d$ be a parameter of the construction, and let $H$ be a $d$-light initialization of $G$.
Then, while there are nodes $s,t \in S$ such that $\dist_H(s,t) > \dist_G(s, t)+4W$, choose any $s \leadsto t$ shortest path $\pi(s,t)$ in $G$ and add all its edges to $H$.
It is immediate that this algorithm terminates with $H$ a $+4W$ subset spanner of $G$, so we now analyze the number of edges $|E_H|$ in the final subgraph $H$.

At any point in the algorithm, say that an ordered pair of nodes $(s, v) \in S \times V$ is \emph{near-connected} if there exists $v'$ adjacent to $v$ in $H$ such that $\dist_H(s, v') = \dist_G(s, v')$.
We then have the following observation
\begin{equation}\label{eqn:dist-sv}
\dist_H(s, v) \le \dist_H(s, v') + W = \dist_G(s, v')+W.
\end{equation}
When nodes $s, t \in S$ with shortest path $\pi(s, t)$ are considered in the construction, there are two cases:
\begin{enumerate}
    \item If there are two nodes $v', v''$ adjacent in $H$ to a node $v \in \pi(s, t)$, and the pairs $(s,v)$ and $(t,v)$ are near-connected, then we have by triangle inequality and~\eqref{eqn:dist-sv}:
    \begin{align*}
        \dist_H(s, t) &\le \dist_H(s, v) + \dist_H(t,v)\\
        &\le (\dist_G(s,v')+W) + (\dist_G(t,v'')+W)\\
        &= \dist_G(s, v') + \dist_G(t,v'') + 2W\\
        &\le \dist_G(s, v) + \dist_G(t,v) + 4W\\
        &= \dist_G(s, t) + 4W.
    \end{align*}
    where the last equality follows from the optimal substructure property of shortest paths. In this case, the path $\pi(s,t)$ is not added to $H$.
    
    \item Otherwise, suppose there is no node $v'$ adjacent in $H$ to a node $v \in \pi(s, t)$ where $(s,v)$ and $(t,v)$ are near-connected.
    After adding the path $\pi(s, t)$ to $H$, every such node $v'$ becomes near-connected to both $s$ and $t$.
    If there are $\ell$ edges in $\pi(s, t)$ currently missing in $H$, then by Theorem \ref{thm:d_initialize} we have $\Omega(\ell d)$ nodes adjacent to $\pi(s, t)$, so $\Omega(\ell d)$ node pairs in $S \times V$ go from not near-connected to near-connected.
    Since there are $\sigma n$ node pairs in $S\times V$, we add a total of $O(\sigma n / d)$ edges to $H$ in this case.
\end{enumerate}

Putting these together, the final size of $H$ is $|E_H| = O\left(nd + \frac{\sigma n}{d}\right)$.
Setting $d := \sqrt{\sigma}$ proves Theorem~\ref{thm:subspan}.
\qed
\end{proof}

We now give our constructions for pairwise spanners. The following lemma will be useful:
\begin{lemma}[\hspace{1sp}\cite{bodwin2020note}]
\label{lem:constant_probability}
Let $a, b > 0$ be absolute constants, and suppose there is an algorithm that, on input $G, P$, produces a subgraph $H$ on $O(n^a |P|^b)$ edges satisfying
$$\dist_H(s, t) \le \dist_G(s, t) + c$$
for at least a constant fraction of the demand pairs $(s, t) \in P$.
Then there is a $+c$ pairwise spanner $H'$ of $G, P$ on $O(n^a |P|^b)$ edges.
\end{lemma}

Using the slack to satisfy only a constant fraction of the demand pairs, we have the following proofs.

\begin{theorem}
\label{theorem:2W_pairwise}
Any graph $G$ with demand pairs $P$ has a $+2W$ pairwise spanner with $O(np^{1/3})$ edges.
\end{theorem}
\begin{proof}
Let $d$ and $\ell$ be parameters of the construction, and let $H$ be a $d$-light initialization of $G$.
For each demand pair $(s, t) \in P$ whose shortest path $\pi(s, t)$ is missing at most $\ell$ edges in $H$, add all edges in $\pi(s, t)$ to $H$.
By Theorem \ref{thm:d_initialize}, any remaining demand pair $(s, t) \in P$ has $\Omega(d\ell)$ nodes adjacent to $\pi(s, t)$.
Let $R$ be a random sample of nodes obtained by including each one independently with probability $1/(\ell d)$; thus, with constant probability or higher, there exists $r \in R$ and $v \in \pi(s, t)$ such that nodes $r$ and $v$ are adjacent in $H$.
Add to $H$ a shortest path tree rooted at each $r \in R$.
We then compute:
\begin{align*}
    \dist_H(s, t) &\leq \dist_H(s, r) + \dist_H(r, t)\\
    &= \dist_G(s, r) + \dist_G(r, t)\\
    &\le \dist_G(s, v) + \dist_G(v, t) + 2W\\
    &= \dist_G(s, t) + 2W.
\end{align*}
The distance for each pair $(s, t) \in P$ is approximately preserved in $H$ with at least a constant probability, which is sufficient for Lemma~\ref{lem:constant_probability}.
The number of edges in the final subgraph $H$ is
$$|E(H)| = O(nd + \ell p + n^2/(\ell d));$$
setting $\ell = n/p^{2/3}$ and $d = p^{1/3}$ proves Theorem~\ref{theorem:2W_pairwise}.
\qed
\end{proof}

\begin{theorem}
\label{theorem:additive_4}
Any graph $G$ with demand pairs $P$ has a $+4W$ pairwise spanner with $O(np^{2/7})$ edges.
\end{theorem}
\begin{proof}
Let $d$ and $\ell$ be parameters of the construction, and let $H$ be a $d$-light initialization of $G$.
For each demand pair $(s, t) \in P$ whose shortest path $\pi(s, t)$ is missing at most $\ell$ edges in $H$, add all edges in $\pi(s, t)$ to $H$.
To handle each $(s, t) \in P$ whose shortest path $\pi(s, t)$ is missing at least $n/d^2$ edges in $H$, we let $R_1$ be a random sample of nodes obtained by including each node independently with probability $d^2/n$, then add a shortest path tree rooted at each $r \in R_1$ to $H$.
By an identical analysis to Theorem \ref{theorem:2W_pairwise}, for each such pair, with constant probability or higher we have
$$\dist_H(s, t) \le \dist_G(s, t) + 2W.$$
Finally, we consider the ``intermediate'' pairs $(s, t) \in P$ whose shortest path $\pi(s, t)$ is missing more than $\ell$ but fewer than $n/d^2$ edges in $H$.
We add the first and last $\ell$ missing edges in $\pi(s, t)$ to the spanner; we will refer to the \emph{prefix} (resp.\ \emph{suffix}) of $\pi(s,t)$ to mean the shortest prefix (suffix) containing these $\ell$ missing edges. By Theorem \ref{thm:d_initialize}, there are $\Omega(\ell d)$ nodes adjacent to the prefix and $\Omega(\ell d)$ nodes adjacent to the suffix.
Let $R_2$ be a random sample of nodes obtained by including each node with probability $1/(\ell d)$, and for each pair $r, r' \in R_2$, add to $H$ all edges in the shortest $r \leadsto r'$ path in $G$ among the paths that are missing at most $n/d^2$ edges (ignore any pair $r, r'$ if no such path exists).
With constant probability or higher, we sample $r, r'$ adjacent to nodes $v, v'$ in the prefix, suffix respectively, in which case we have:
\begin{align*}
        \dist_H(s, t) &\leq \dist_H(s, v) + \dist_H(v, v') + \dist_H(v', t)\\
        &= \dist_G(s, v) + \dist_H(v, v') + \dist_G(v', t)\\
        &\le \dist_G(s, v) + \dist_H(r, r') + 2W + \dist_G(v', t).
\end{align*}
Notice that $\dist_H(r, r') \le 2W + \dist_G(v, v')$, due to the existence of the path $r \circ \pi(s, t)[v, v'] \circ r'$ which is indeed missing $\le n/d^2$ edges.
Thus we may continue:
\begin{align*}
        &\le \dist_G(s, v) + \dist_G(v, v') + 4W + \dist_G(v', t)\\
        & = \dist_G(s, t) + 4W.
\end{align*}

The distance for each pair $(s, t) \in P$ is approximately preserved in $H$ with at least constant probability, which again suffices by Lemma \ref{lem:constant_probability}, and the number of edges in $H$ is
$$|E(H)| = O\left(nd + p\ell + n^3/(\ell^2d^4) \right).$$
Setting $\ell = n/p^{5/7}$ and $d = p^{2/7}$ completes the proof of Theorem~\ref{theorem:additive_4}.
\qed
\end{proof}

\begin{theorem}
\label{theorem:additive_6}
Any graph $G$ with demand pairs $P$ has a $+8W$ pairwise spanner containing $O(np^{1/4})$ edges.
\end{theorem}
\begin{proof}
Let $\ell,d$ be parameters of the construction and let $H$ be a $d$-light initialization of $G$.
For each $(s, t) \in P$ whose shortest path $\pi(s, t)$ is missing $\le \ell$ edges in $H$, add all edges in $\pi(s, t)$ to $H$.
Otherwise, like before, we add the first and last $\ell$ missing edges of $\pi(s, t)$ to $H$ (prefix and suffix).
Then, randomly sample a set $R$ by including each node with probability $1/(\ell d)$, and use Theorem~\ref{thm:subspan} to add a $+4W$ subset spanner on the nodes in $R$.
By Theorem~\ref{thm:d_initialize}, the prefix and suffix each have $\Omega(\ell d)$ adjacent nodes.
Thus, with constant probability or higher, we sample $r, r' \in R$ adjacent to $v, v'$ in the added prefix and suffix respectively.
We then compute:
\begin{align*}
    \dist_H(s, t) & \leq \dist_H(s, v) + \dist_H(v, v') + \dist_H(v', t)\\
    &\le \dist_G(s, v) + \dist_H(v, v') + \dist_G(v', t)\\
    &\le \dist_G(s, v) + \dist_H(r, r') + 2W + \dist_G(v', t)\\
    &\le \dist_G(s, v) + \dist_G(r, r') + 6W + \dist_G(v', t)\\
    &\le \dist_G(s, v) + \dist_G(v, v') + 8W + \dist_G(v', t)\\
    &= \dist_G(s, t) + 8W.
\end{align*}
Again, the distance for each pair $(s, t) \in P$ is approximately preserved in $H$ with at least constant probability, which suffices by Lemma~\ref{lem:constant_probability}.
The number of edges in $H$ is
$$|E(H)|=O\left(nd + p\ell + n^{3/2}/\sqrt{\ell d}\right).$$
Setting $\ell = n/p^{3/4}$ and $d = p^{1/4}$ completes the proof of Theorem~\ref{theorem:additive_6}.
\qed
\end{proof}

\section{All-pairs Additive Spanners}
\label{section:all_pairs}

We now turn to the \emph{all-pairs} setting, i.e., demand pairs $P = V \times V$.
These very closely follow the associated constructions in the unweighted setting.

\begin{theorem} \label{thm:2W_allpairs}
Every $n$-node graph has a $+2W$ additive spanner on $O(n^{3/2})$ edges.
\end{theorem}
\begin{proof}
There are two steps in the construction of the spanner $H$.
Initially $H$ is empty (not an initialization).
First, we take a random sample of nodes $R$ by including each one independently with probability $1/d$, and we add to $H$ a shortest path tree rooted at each $r \in R$.
Then, a new step: say that a vertex $v$ is \emph{hit} if we sample $r \in R$ adjacent to $v$ in the original input graph $G$, or it is \emph{missed} otherwise.
For each missed vertex $v$, add all edges incident to $v$ to the spanner.

First let us count the number of edges in the final spanner $H$.
We have $O(n^2/d)$ edges (in expectation) from the random BFS trees.
A node of degree $\Delta \ge d$ is hit with probability $\ge \Theta(d/\Delta)$, and thus it costs $O(d)$ edges in expectation in the final step, where we add $\Delta$ edges if it is missed.
So this costs an additional $O(nd)$ edges in the spanner, in expectation.
The claimed size bound of $O(n^{3/2})$ then follows by setting $d = n^{1/2}$.

The proof of correctness of the spanner is essentially the same as in Theorem \ref{theorem:2W_pairwise}.
Consider a node pair $s, t$ and a shortest path $\pi(s, t)$ between them.
If all edges in $\pi(s, t)$ are present in the spanner $H$, then we have $\dist_G(s, t) = \dist_H(s, t)$.
Otherwise, let $(u, v)$ be a missing edge in $\pi(s, t)$.
Thus $u$ is hit, since we did not add all of its incident edges to the spanner.
So we sampled a node $r \in R$ adjacent to $u \in \pi(s, t)$.
It now follows from exactly the same logic as Theorem \ref{theorem:2W_pairwise} that $\dist_H(s, t) \le \dist_G(s, t) + 2W$.
\end{proof}

\begin{theorem}
Every $n$-node graph has a $+4W$ spanner on $\Oish(n^{7/5})$ edges.
\end{theorem}
\begin{proof}
The construction is essentially the same as Theorem \ref{theorem:additive_4}, except that the parameter $\ell$ no longer plays a role, and we skip the step where we add shortest paths $\pi(s, t)$ that are missing $\le \ell$ edges in the initialized spanner.
In total, the construction is:
\begin{itemize}
    \item Start the spanner $H$ as a $d$-light initialization of the input graph $G$,
    \item Let $R_1$ be a random sample of nodes obtained by including each node independently with probability $d^2/n$, then add a shortest path tree rooted at each $r \in R_1$ to $H$, and
    \item Let $R_2$ be a random sample of nodes obtained by including each node with probability $1/d$.  For each pair of nodes $r, r' \in R_2$, add to $H$ the shortest $r \leadsto r'$ path among the paths that are missing $\le n/d^2$ edges in the initialization (ignore the pair $r, r'$ if no such path exists).
\end{itemize}
This construction has $O(nd + n^3 / d^4)$ edges in total, which is $O(n^{7/5})$ by setting $d = n^{2/5}$.
By the same logic as in Theorem \ref{theorem:additive_4}, each demand pair $(s, t) \in V \times V$ is satisfied in $H$ with constant probability.
Thus, if we repeat the construction $C \log n$ times for a large enough absolute constant $C$, then at the end every demand pair is satisfied with high probability.
Unioning together the spanners computed in each round, the final size is $O(n^{7/5} \cdot \log n) = \Oish(n^{7/5})$.
\end{proof}

\begin{theorem}
Every $n$-node graph has a $+8W$ additive spanner on $O(n^{4/3})$ edges.
\end{theorem}
\begin{proof}
Initially, the spanner $H$ is empty.
We take a random sample of nodes $R$ by including each one independently with probability $1/d$, and we add a $+4W$ subset spanner on $R \times R$.
Then, we again say that a vertex $v$ is \emph{hit} if we sample $r \in R$ adjacent to $v$ in the original input graph $G$, or it is \emph{missed} otherwise.
For each missed vertex $v$, add all edges incident to $v$.
For each hit vertex $v$, add an edge $(v, r)$ to one node $r \in R$.

The size of the spanner is $O(nd)$ in expectation from the missed verticest (by the same argument as in Theorem \ref{thm:2W_allpairs}), and $O(n(n/d)^{1/2}) = O(n^{3/2} / d^{1/2})$ in expectation from the subset spanner.
The claimed size bound follows by setting $d = n^{1/3}$.

Finally, correctness follows by basically the same argument as in Theorem \ref{theorem:additive_6}.
Consider a node pair $s, t$, and let $\pi(s, t)$ be a shortest path.
If all edges in $\pi(s, t)$ are present in the spanner $H$, then we have $\dist_H(s, t) = \dist_G(s, t)$.
Otherwise, let $(u, v)$ be the first missing edge on $\pi(s, t)$ and let $(x, y)$ be the last missing edge on $\pi(s, t)$.
Thus $u, y$ are both hit, since we did not add all their incident edges to the spanner.
So there are nodes $r, r' \in R$ adjacent to $u, t$ in the spanner $H$.
Using the $+4W$ subset spanner, this implies that $\dist_H(s, t) \le \dist_G(s, t) + 8W$, by the same calculation as in Theorem \ref{theorem:additive_6}.
\end{proof}

\section{Conclusions and Open Problems}

We have shown that most important unweighted additive spanner constructions have natural weighted analogues.
At present, the exceptions are the $+4W$ subset spanner on $O(n |S|^{1/2})$ edges (which should probably have only $+2W$ error) and the $+8W$ all-pairs/pairwise spanners (which should probably have only $+6W$ error).
Closing these error gaps is an interesting open problem.
It would also be interesting to obtain weighted analogues of related concepts, most notably, the Thorup-Zwick emulators~\cite{TZ06}, which are optimal~\cite{ABP17} in essentially the same way that the 6-additive spanner on $O(n^{4/3})$ edges is optimal.

Finally, as mentioned earlier, it would be interesting to find constructions of \emph{purely} additive spanners parametrized by some other statistic besides the maximum edge weight $W$; a natural parameter is $W(u,v)$, the maximum edge weight along a shortest $u$-$v$ path.


\bibliography{references}
\end{document}